\documentclass[conference]{IEEEtran}
\IEEEoverridecommandlockouts

\usepackage{cite}
\usepackage{amsmath,amssymb,amsfonts}
\usepackage{algorithm}
\usepackage{algpseudocode}%
\usepackage{graphicx}
\usepackage{hyperref}
\usepackage{textcomp}
\usepackage{xcolor}
\usepackage{subcaption}
\usepackage{multirow}
\usepackage{siunitx}

\usepackage{tabularx}
\makeatletter
\newcommand{\multiline}[1]{%
  \begin{tabularx}{\dimexpr\linewidth-\ALG@thistlm}[t]{@{}X@{}}
    #1
  \end{tabularx}
}
\makeatother

\usepackage{lipsum}
\usepackage{mathtools}
\usepackage{cuted}
\usepackage{stfloats}
\usepackage{bbm} 

\usepackage{enumitem}

\usepackage{amsthm}
\usepackage{balance} 

\newtheorem{remark}{Remark}
\newtheorem{proposition}{Proposition}
\newtheorem{conjecture}{Conjecture}

\newenvironment{proofN3}{%
  \proof}{\endproof}

\usepackage{amsmath,bm}

\begin{document}

\title{
 Linear Coding for Gaussian Two-Way Channels
}
\author{\IEEEauthorblockN{Junghoon Kim\IEEEauthorrefmark{1}, Seyyedali Hosseinalipour\IEEEauthorrefmark{1},
Taejoon Kim\IEEEauthorrefmark{2},
David J. Love\IEEEauthorrefmark{1} 
and Christopher G. Brinton\IEEEauthorrefmark{1}}
\IEEEauthorblockA{\IEEEauthorrefmark{1}Electrical and Computer Engineering, Purdue University, West Lafayette, IN, USA}
\IEEEauthorblockA{\IEEEauthorrefmark{2}Electrical Engineering and Computer Science, University of Kansas, Lawrence, KS, USA}
\IEEEauthorblockA{\IEEEauthorrefmark{1}\{kim3220, hosseina, djlove, cgb\}@purdue.edu, 
\IEEEauthorrefmark{2}taejoonkim@ku.edu}
\thanks{This work was supported by the Office of Naval Research (ONR) under Grants N00014-21-1-2472 and N00014-22-1-2305, and
the National Science Foundation (NSF) under Grants CNS 2225577 and 2212565.}
}
\maketitle

\begin{abstract}

We consider linear coding for Gaussian two-way channels (GTWCs), in which
each user generates the transmit symbols by linearly encoding both its message and the past received symbols (i.e., the feedback information) from the other user.
In Gaussian one-way channels (GOWCs), Butman has proposed a well-developed model for linear encoding that encapsulates feedback information into transmit signals. 
However,
such a model for GTWCs has not been well studied since the coupling of the encoding
processes at the users in GTWCs renders the encoding design non-trivial and challenging.
In this paper, we aim to fill this gap in the literature by extending the existing signal models in GOWCs to GTWCs.
With our developed signal model for GTWCs,
 we 
 formulate an optimization problem to jointly design the encoding/decoding schemes for both the users, aiming to minimize the weighted sum of their transmit powers under signal-to-noise ratio constraints.
First, we derive an optimal form of the linear decoding schemes under any arbitrary encoding schemes employed at the users.
Further, we 
provide new insights on the encoding design for GTWCs.
 In particular, we show that it is optimal that one of the users (i) does not transmit the feedback information to the other user at the last channel use, and (ii) transmits its message only over the last channel use. 
 With these solution behaviors, we further simplify the problem and solve it via an iterative two-way optimization scheme.
We numerically demonstrate that our proposed scheme for GTWCs achieves a better performance 
in terms of the transmit power 
compared to the existing counterparts, such as the non-feedback scheme and  one-way optimization scheme.

\end{abstract}

\section{Introduction}

The two-way channel was first studied by 
Shannon \cite{shannon1961two}, where two users exchange their messages with each other through their separate channels.
In this paper, we consider Gaussian two-way channels (GTWCs), where  Gaussian  noise is added independently to each way of the channels between the users.
Han in \cite{han1984general} showed that 
incorporating feedback information (i.e., the previously received symbols) into
transmit symbols for encoding does not increase 
the capacity of GTWCs. Nevertheless, it has been shown that feedback can improve the communication reliability of Gaussian channels~\cite{schalkwijk1966coding,chance2011concatenated,agrawal2011iteratively,kim2020deepcode,palacio2021achievable,mishra2021linear,butman1969general}.

For Gaussian one-way channels (GOWCs), the seminal work done by Schalkwijk and Kailath (S-K) in \cite{schalkwijk1966coding} introduced a simple linear encoding that can achieve doubly exponential decay in the probability of error upon having noiseless feedback information. 
In \cite{chance2011concatenated}, Chance and Love~proposed a linear encoding scheme for GOWCs with noisy feedback, which is further analyzed and revealed to be the optimal linear encoding scheme by~\cite{agrawal2011iteratively}.
In \cite{kim2020deepcode}, Kim. \textit{et al.} exploited deep learning for a non-linear coding in GOWCs and showed performance improvements in the error probability.

For GTWCs, several recent works have revealed the advantages of feedback in terms of improving  communication reliability.
In \cite{palacio2021achievable}, Palacio-Baus and Devroye
showed that feedback can improve the error exponent as compared to the non-feedback case.
In \cite{vasaldynamic}, Vasal suggested a dynamic programming (DP)-based methodology for encoding in GTWCs.
Although the effectiveness of the DP approach in GTWCs has not been verified, the author's previous work \cite{mishra2021linear} 
revealed
that the DP approach is 
effective in GOWCs with noisy feedback.

To the best of our knowledge, 
a general system model for linear encoding in GTWCs
has not been well studied,
unlike the well-developed counterpart for GOWCs proposed by Butman~\cite{butman1969general}. Furthermore, designing the linear encoding schemes for GTWCs is a non-trivial process since the \textit{coupling} of the encoding processes at the users should be encapsulated in the system model. 
In this paper, we aim to bridge the gaps between the two pieces of literature on GOWCs and GTWCs. To this end, we propose a general system model for linear coding in GTWCs by extending the existing formulations in GOWC literature~\cite{butman1969general, chance2011concatenated, agrawal2011iteratively} to GTWCs.

Furthermore, using our developed signal model for GTWCs, we  define the signal-to-noise
ratio (SNR) at the users, and then
derive an optimal form of the linear decoding schemes 
by maximizing the SNRs
under arbitrary encoding schemes employed at the users.
We then formulate the weighted sum transmit power minimization problem to satisfy arbitrary SNR thresholds,
aiming to jointly optimize the encoding/decoding schemes of the users.
To mitigate the coupling effect caused by encoding processes at the users, we assume that one of the users (i.e., User 2) feeds back only recently received signal.
Under this assumption, we theoretically characterize the optimal solution for a part of the
encoding schemes.
In particular, we first prove that it is optimal for one of the users (i.e., User 2) not to utilize the last channel use for feeding back the previously received signals to  the other user (i.e., User 1).
Second, based on our conjecture, we claim that it is optimal for User 2 to transmit the message only over the last channel use. 
From our theoretical insights on the encoding and decoding design,
we further simplify the optimization problem and propose an iterative
two-way optimization scheme to solve it.
Through numerical experiments, we reveal that our proposed two-way optimization scheme 
outperforms the open loop (i.e., non-feedback) and the one-way optimization schemes.


\section{System Model in Gaussian Two-way Channels}


\begin{figure}[t]
    \centering
    \includegraphics[width=\linewidth]{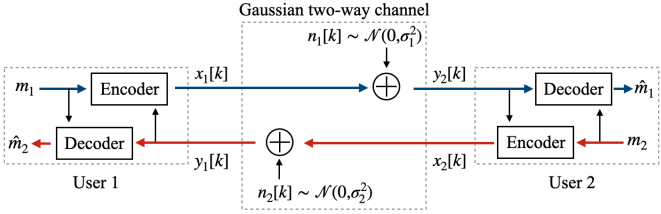}
    \caption{System model for Gaussian two-way channels.}
    \label{fig:system}
\end{figure}

We consider a two-way channel between two users, User 1 and User 2, as shown in Fig. \ref{fig:system}. We assume that User 1 and 2 perform linear encoding and decoding of blocklength $N$.
Let $k\in\{1,\cdots, N\}$ denote the index of channel use, and $x_1[k]$ and $x_2[k]$ represent the transmit signals at User 1 and User 2, respectively, at time $k$.
We consider additive white Gaussian noise (AWGN) channels between the users. Subsequently,
the received signal at User 2, ${y}_2[k]$, and User 1, ${y}_1[k] $, at time $k$ are given by
\begin{align}
    {y}_2[k] & = x_1[k] + n_1[k],
    \\
    {y}_1[k] & = x_2[k] + {n}_2[k],
\end{align}
respectively, 
where ${n}_1[k] \sim \mathcal{N}(0,\sigma_1^2)$ and ${n}_2[k] \sim \mathcal{N}(0,\sigma_2^2)$ are Gaussian noises. 
Considering signal exchange over the blocklength of $N$, we represent the received signals at User 2 and User 1 in vector form as ${\bf y}_2 = [y_2[1], ..., y_2[N]]^\top \in \mathbb{R}^{N \times 1}$ and ${\bf y}_1 = [y_1[1], ..., y_1[N]]^\top \in \mathbb{R}^{N \times 1}$, respectively, given by
\begin{align}
    {\bf y}_2 & = {\bf x}_1 + {\bf n}_1,
    \label{eq:y2}
    \\
    {\bf y}_1 & = {\bf x}_2 + {\bf n}_2,
    \label{eq:y1}
\end{align}
where ${\bf x}_i = [x_i[1], ..., x_i[N]]^\top$ and ${\bf n}_i = [n_i[1], ..., n_i[N]]^\top$, $i\in\{1,2\}$.

The goal of signal exchange among the users is to transmit the \textit{message} available at each user to the other. In particular, each User $i$, $i\in\{1,2\}$, aims to transmit a unique message $m_i \in \mathbb{R}$ to the other user where  $\mathbb{E}[m_i] = 0$ and $\mathbb{E}[\vert m_i \vert^2] = 1$.
%
Motivated by the advantages of incorporating the received signals into transmit signals through feedback in GOWCs,
e.g., enhancing communication reliability~\cite{agrawal2011iteratively,butman1969general,chance2011concatenated, schalkwijk1966coding}, 
we consider a linear coding framework at the users that exploits the feedback in GTWCs.
In our framework, users receive some feedback information
from one another and utilize that to generate their transmit signals.
User $i$
constructs the transmit signal at time $k$, $x_i[k]$, as a function of the message $m_i$ and the received signals up to time $k-1$,~$\{ {y}_i[\tau]\}_{\tau=1}^{k-1}$.

We consider that User $i$, $i \in \{1,2\}$, employs the message encoding vector  $\tilde {\bf g}_i \in \mathbb{R}^{N \times 1}$ for encoding the message $m_i$ and the feedback encoding matrix $\tilde {\bf F}_i \in \mathbb{R}^{N \times N}$ for encoding the received signals.
Note that $\tilde {\bf F}_i$, $i \in \{1,2\}$, is strictly lower triangular (i.e., the matrix entries are zero on and above the diagonal) due to causality of the system.
To avoid feeding back redundant information,
 we consider that each user removes the contribution of its known prior transmitted signals  from the received signals to generate its future transmit signals.
For the case of User 1, the transmit signal ${\bf x}_1$ is encoded by User 2 with $\tilde {\bf F}_2$ and then transmitted back to User 1. Therefore,
 User 1 subtracts its signal portion ${\bf x}_1$ from the receive signal ${\bf y}_1$ as ${\bf y}_1 - \tilde {\bf F}_2 {\bf x}_1$.
It is obvious that using the feedback information ${\bf y}_1$ is equivalent to using the modified feedback information ${\bf y}_1 - \tilde {\bf F}_2 {\bf x}_1$.
Similarly, User 2 subtracts its signal portion ${\bf x}_2$ from ${\bf y}_2$ and obtains the modified feedback information ${\bf y}_2 - \tilde {\bf F}_1 {\bf x}_2$. 
The transmit signals of the users are then given by
\begin{align}
    {\bf x}_1 & = \tilde {\bf g}_1 m_1 + \tilde {\bf F}_1 ({\bf y}_1 - \tilde {\bf F}_2 {\bf x}_1),
    \label{eq:x1:tilde}
    \\
    {\bf x}_2 & = \tilde {\bf g}_2 m_2 + \tilde {\bf F}_2 ({\bf y}_2 - \tilde {\bf F}_1 {\bf x}_2).
    \label{eq:x2:tilde}
\end{align}


Since each of the users transmits the signals encapsulating the received signals   from the other over the multiple channel uses, a \textit{coupling} occurs between the transmit signals at the users. 
To mitigate the coupling effects in the signal representation,
we rewrite the signal model in
\eqref{eq:x1:tilde}-\eqref{eq:x2:tilde} as
\begin{align}
    {\bf x}_1 & = {\bf g}_1 m_1 +  {\bf F}_1 ({\bf y}_1 - {\bf F}_2 {\bf x}_1),
    \label{eq:x1}
    \\
    {\bf x}_2 & = {\bf g}_2 m_2 + {\bf F}_2 {\bf y}_2,
    \label{eq:x2}
\end{align}
by expressing ${\bf g}_1$, ${\bf F}_1$, ${\bf g}_2$, and ${\bf F}_2$ as functions of $\tilde {\bf g}_1$, $\tilde  {\bf F}_1$, $\tilde  {\bf g}_2$, and $\tilde  {\bf F}_2$.
Specifically, we can reformulate the equation in \eqref{eq:x2:tilde} and obtain 
\begin{align}
    {\bf x}_2 = ({\bf I} + \tilde {\bf F}_2 \tilde {\bf F}_1)^{-1}\tilde {\bf g}_2 m_2 + ({\bf I} + \tilde {\bf F}_2 \tilde {\bf F}_1)^{-1}\tilde {\bf F}_2 {\bf y}_2.
    \label{eq:x2:reform}
\end{align}
By comparing the equations in \eqref{eq:x2} and \eqref{eq:x2:reform}, we can find ${\bf g}_2 = ({\bf I} + \tilde {\bf F}_2 \tilde {\bf F}_1)^{-1} \tilde {\bf g}_2$ and ${\bf F}_2 = ({\bf I} + \tilde {\bf F}_2 \tilde {\bf F}_1)^{-1}\tilde {\bf F}_2$. Similarly, we can rewrite the equation in \eqref{eq:x1:tilde} as
\begin{align}
    {\bf x}_1 & =  \tilde {\bf g}_1 m_1 + \tilde {\bf F}_1 ({\bf y}_1 - \tilde {\bf F}_2 {\bf x}_1 - {\bf F}_2 {\bf x}_1 +  {\bf F}_2 {\bf x}_1)
    \nonumber
    \\
    & = \big({\bf I} - \tilde {\bf F}_1 ( {\bf F}_2 - \tilde {\bf F}_2) \big)^{-1} \tilde {\bf g}_1 m_1
    \nonumber
    \\
    & \hspace{1cm}
    + \big({\bf I} - \tilde {\bf F}_1 ( {\bf F}_2 - \tilde {\bf F}_2) \big)^{-1} \tilde {\bf F}_1 ({\bf y}_1  - {\bf F}_2 {\bf x}_1).
    \label{eq:x1:reform}
\end{align}
By comparing the equations in \eqref{eq:x1} and \eqref{eq:x1:reform}, we can find ${\bf g}_1 = {\bf A}^{-1} \tilde {\bf g}_1$ and ${\bf F}_1 = {\bf A}^{-1} \tilde {\bf F}_1$, where ${\bf A} 
= {\bf I} - \tilde {\bf F}_1 ( {\bf F}_2 - \tilde {\bf F}_2) 
= {\bf I} - \tilde {\bf F}_1 ( ({\bf I} + \tilde {\bf F}_2 \tilde {\bf F}_1)^{-1} - {\bf I}) \tilde {\bf F}_2$.
Note that both ${\bf F}_1$ and ${\bf F}_2$ are  strictly lower triangular.

Henceforth, we aim to design ${\bf g}_1$, ${\bf F}_1$, ${\bf g}_2$, and ${\bf F}_2$ and focus on the signal representation in \eqref{eq:x1}-\eqref{eq:x2}.
Accordingly, we rewrite the received signal expressions in \eqref{eq:y2}-\eqref{eq:y1} as
\begin{align}
    {\bf y}_1 & = {\bf g}_2 m_2 + {\bf F}_2 {\bf y}_2  + {\bf n}_2,
    \label{eq:y1:vec}
    \\
    {\bf y}_2 
    & = {\bf g}_1 m_1 +  {\bf F}_1 ({\bf y}_1 - {\bf F}_2 {\bf x}_1) + {\bf n}_1
    \nonumber
    \\
    & = {\bf g}_1 m_1  + {\bf F}_1 {\bf g}_2 m_2 +
    ({\bf I} + {\bf F}_1{\bf F}_2){\bf n}_1 + {\bf F}_1{\bf n}_2.
    \label{eq:y2:vec}
\end{align}

Considering the received signals \eqref{eq:y1:vec}-\eqref{eq:y2:vec},  the transmit signals in  \eqref{eq:x1}-\eqref{eq:x2} can be written as the sum of the messages and noises as follows:
\begin{align}
    {\bf x}_1 
    & = {\bf g}_1 m_1  + {\bf F}_1 ({\bf g}_2 m_2 +
    {\bf F}_2{\bf n}_1 +{\bf n}_2),
    \\
    {\bf x}_2
    & = {\bf g}_2 m_2 + {\bf F}_2 ({\bf g}_1 m_1  + {\bf F}_1 {\bf g}_2 m_2 +
    ({\bf I} + {\bf F}_1{\bf F}_2){\bf n}_1 + {\bf F}_1{\bf n}_2)
    \nonumber
    \\
    & = ({\bf I} + {\bf F}_2 {\bf F}_1 ) {\bf g}_2 m_2 +  {\bf F}_2{\bf g}_1 m_1  
    \nonumber
    \\
    & \hspace{2.65cm} 
    + {\bf F}_2({\bf I} + {\bf F}_1{\bf F}_2){\bf n}_1 + {\bf F}_2{\bf F}_1{\bf n}_2.
\end{align}
Using the above two expressions, we formulate the transmit power of the users as
\begin{align}
    \mathbb{E}\big[\|{\bf x}_1\|^2\big] 
    &= \|{\bf g}_1\|^2 +  \|{\bf F}_1 {\bf g}_2 \|^2 +  \|{\bf F}_1 {\bf F}_2 \|_F^2  \sigma_1^2
    + \| {\bf F}_1 \|_F^2  \sigma_2^2,
    \label{eq:power1}
    \\
    \mathbb{E} \big[\|{\bf x}_2\|^2\big] 
    &= \| ({\bf I} + {\bf F}_2 {\bf F}_1 ) {\bf g}_2\|^2 + \|{\bf F}_2{\bf g}_1\|^2
    \nonumber
    \\
    & \hspace{.7cm}
    + \| {\bf F}_2( {\bf I} + {\bf F}_1{\bf F}_2 ) \|_F^2 \sigma_1^2
    + \| {\bf F}_2{\bf F}_1 \|_F^2 \sigma_2^2, 
    \label{eq:power2}
\end{align}
where the messages and the noises are assumed to be uncorrelated to each other. 
In the following section, we define SNRs of the users and obtain the optimal decoding schemes by maximizing the SNRs.
We then introduce our encoding design and solution method.


\section{Linear Encoding and Decoding Schemes in GTWC}

\subsection{Design of Optimal Linear Decoding Schemes}
\label{ssec:decoding}

Since the decoding is conducted at each of the users independently, we can use the same technique used in GOWC to find the optimal linear decoding scheme~\cite{chance2011concatenated, agrawal2011iteratively}.
After the $N$ channel uses, 
each user aims to estimate the message of the other user.
We first consider that User 1 estimates $m_2$ with the received signal ${\bf y}_1$ by using a linear combining vector ${\bf w}_{2} \in \mathbb{R}^{N \times 1}$.
By plugging ${\bf y}_2$ (given by \eqref{eq:y2:vec}) in ${\bf y}_1$ (given by \eqref{eq:y1:vec}), we can
rewrite ${\bf y}_1$ as ${\bf y}_1 = {\bf F}_2 {\bf g}_1 m_1 + ({\bf I} + {\bf F}_2 {\bf F}_1){\bf g}_2 m_2 +  {\bf F}_2 ({\bf I} + {\bf F}_1 {\bf F}_2) {\bf n}_1  + ({\bf I} + {\bf F}_2 {\bf F}_1){\bf n}_2$.
Through a pre-processing phase, User 1 
is assumed to 
subtract its message contribution, ${\bf F}_2 {\bf g}_1 m_1$, from ${\bf y}_1$ 
to obtain 
${  {\bf z}_1} = {\bf y}_1 - {\bf F}_2 {\bf g}_1 m_1 = ({\bf I} + {\bf F}_2 {\bf F}_1){\bf g}_2 m_2 +  ({\bf I} + {\bf F}_2 {\bf F}_1) {\bf F}_2 {\bf n}_1  + ({\bf I} + {\bf F}_2 {\bf F}_1){\bf n}_2$.
For estimating $m_1$, using ${ {\bf z}_1}$ is equivalent to using
${ \tilde {\bf y}_1} =  ({\bf I} + {\bf F}_2 {\bf F}_1)^{-1}  {\bf z}_1  =  {\bf g}_2 m_2 + {\bf F}_2 {\bf n}_1  + {\bf n}_2$. 
Using the result of pre-processing,
User 1 obtains the message estimate $ {\hat m}_2 = {\bf w}_{2}^\top \tilde {\bf y}_1$. 
The SNR used
to estimate $m_2$ is
\begin{align}
    {\rm SNR}_2 = \frac{\vert {\bf w}_{2}^\top {\bf g}_2 \vert^2} 
    {{\bf w}_{2}^\top {\bf Q}_2 {\bf w}_{2}},
    \label{eq:SNR2}
\end{align}
where 
\begin{align}
    {\bf Q}_2 & = {\bf F}_2 {\bf F}_2^\top \sigma_1^2 + \sigma_2^2 {\bf I}.
    \label{eq:Q2}
\end{align}

Similarly, we consider that User 2 estimates $m_1$ using the received signal ${\bf y}_2$ via
a linear combining vector ${\bf w}_{1} \in \mathbb{R}^{N \times 1}$. In pre-processing, 
User 2 is assumed to subtract its message contribution from
${\bf y}_2$ in \eqref{eq:y2:vec}, and obtains 
$\tilde {\bf y}_2 =  {\bf y}_2 - {\bf F}_1 {\bf g}_2 m_2 = {\bf g}_1 m_1  + ({\bf I} + {\bf F}_1{\bf F}_2){\bf n}_1 + {\bf F}_1{\bf n}_2$.
User 2 then obtains the message estimate $ {\hat m}_1 = {\bf w}_{1}^\top \tilde {\bf y}_2$ under SNR 
\begin{align}
    {\rm SNR}_1 = \frac{\vert {\bf w}_{1}^\top {\bf g}_1 \vert^2} 
    {{\bf w}_{1}^\top {\bf Q}_1 {\bf w}_{1}},
    \label{eq:SNR1}
\end{align}
where 
\begin{align}
    {\bf Q}_1 & = ({\bf I} + {\bf F}_1{\bf F}_2) ({\bf I} + {\bf F}_1{\bf F}_2)^\top \sigma_1^2 
    + {\bf F}_1 {\bf F}_1^\top \sigma_2^2.
    \label{eq:Q1}
\end{align}
%
Given ${\bf g}_1$, ${\bf F}_1$, ${\bf g}_2$, and ${\bf F}_2$, 
the optimal combining vector (that minimizes the error probability for message estimation) is obtained by maximizing the SNRs
given by  \cite{chance2011concatenated, agrawal2011iteratively}
\begin{align}
    {\bf w}^\star_{i} = \frac{{\bf Q}_i^{-1} {\bf g}_i}
    { {\bf g}_i^\top {{\bf Q}_i}^{-1} {\bf g}_i}, \quad i\in\{1,2\}.
    \label{eq:w_opt}
\end{align}
Plugging \eqref{eq:w_opt} in \eqref{eq:SNR2} and \eqref{eq:SNR1}, we obtain the SNR as
\begin{align}
      {\rm SNR}_i =  {\bf g}_i^\top {\bf Q}_i^{-1} {\bf g}_i, \quad i \in \{1,2\}.
      \label{eq:SNR_opt}
\end{align}

\subsection{Optimization Formulation for Linear Encoding Schemes}
\label{ssec:encoding}

The decoding schemes 
in \eqref{eq:w_opt} are represented as functions of the encoding schemes
of the users.
Thus, 
for joint encoding and decoding design, we focus on designing the encoding schemes 
with
the derived form of SNRs in \eqref{eq:SNR_opt}.
In this work, we
minimize the weighted sum of the users' transmit powers  under their SNR constraints. Accordingly, the optimization problem is given by
\begin{align}
  (\bm{\mathcal{P}}): \hspace{0.2cm} & \underset{{\bf g}_1, {\bf F}_1, {\bf g}_2, {\bf F}_2}
    {\text{min}} & & 
 \alpha \mathbb{E}\big[\|{\bf x}_1 \|^2\big] + (1-\alpha) \mathbb{E}\big[\|{\bf x}_2 \|^2\big]
\nonumber
\\
& \text{s.t.}
& &  
{\rm SNR}_1 = \eta_1, \quad
{\rm SNR}_2 = \eta_2,
\label{opt:con}
\end{align}
where $\eta_1,\eta_2\in\mathbb{R}^+$ are the target SNRs and $\alpha\in(0,1)$ is a weighting coefficient.

\begin{remark}
The equality constraints in \eqref{opt:con} are equivalent to inequality constraints ${\rm SNR}_1 \ge \eta_1$ and ${\rm SNR}_2 \ge \eta_2$ in terms of minimizing the objective function of $\bm{\mathcal{P}}$.
This is because if we obtain ${\bf g}_1$ such that ${\rm SNR}_1 > \eta_1$, we can always  choose $\bar {\bf g}_1 = (1-\epsilon){\bf g}_1$ with $\epsilon \in (0,1)$ under which $ {\rm SNR}_1 =  \bar {\bf g}_1^\top {\bf Q}_1^{-1} \bar  {\bf g}_1 = \eta_1$. This choice of $\bar {\bf g}_1$ will yield a smaller transmit power of the users in \eqref{eq:power1} and \eqref{eq:power2}, since $\| \bar {\bf g}_1 \|^2 < \| {\bf g}_1 \|^2$ and $\| {\bf F}_2 \bar {\bf g}_1 \|^2 < \| {\bf F}_2 {\bf g}_1 \|^2$. The same holds for the value of ${\rm SNR}_2$.
\end{remark}

In general, at time $k \ge 2$, 
User $i$
feeds back a linear combination of the previously received signals up to time $k-1$, i.e., $\{ {y}_i[\tau]\}_{\tau=1}^{k-1}$, where
$i \in \{1,2\}$.
This implies that the initially received signals at the users are repetitively fed back to the other over a total of $N$ channel uses, e.g., the information of $y_2[1]$ at User 2 is fed back to User 1 over $N-1$ times.
This repetitive feedback in both ways would make the design of the encoding schemes more complicated because the encoding schemes of the users are coupled.
To mitigate the complexity of designing the encoding schemes, 
we assume that 
User 2 only feeds back the recently received signal of ${\bf y}_2$ in \eqref{eq:x2}, i.e., ${\bf F}_2$ is in the form of
\begin{equation}
{\bf F}_2 = 
\begin{bmatrix}
0 & 0 & 0 & ... & 0\\
f_{2,2} & 0 & 0 & ... & 0 \\
0 & f_{2,3} & 0 & ... & 0  \\
\vdots & \vdots & \ddots & & 0 \\
0 & 0 & ... & f_{2,N} & 0
\end{bmatrix} 
\in \mathbb{R}^{N \times N}.
\label{eq:F2}
\end{equation}

First, we investigate the solution behavior for the feedback of User 2.
Specifically, we reveal
that it is optimal for User 2  not to utilize the last channel use for feeding back the previous received signals to User 1, i.e., $f_{2,N}=0$, for arbitrary encoding schemes.

\begin{proposition}
\label{pro1}
In the problem $\bm{\mathcal{P}}$ with ${\bf F}_2$ in the form of \eqref{eq:F2},
it is optimal that $f_{2,N} = 0$.
\end{proposition}
\begin{proof}
We let (${\bf g}_1$, ${\bf F}_1$, ${\bf g}_2$, ${\bf F}_2$) be any feasible solution to $\bm{\mathcal{P}}$.
We also let $\bar {\bf F}_2$ be equal to ${\bf F}_2$, except that the last entry of $\bar {\bf F}_2$ is zero, i.e.,  $\bar f_{2,N}=0$. 
We will show that  (i) the solution (${\bf g}_1$, ${\bf F}_1$, $\bar {\bf g}_2$, $\bar {\bf F}_2$)  is a feasible solution where $\bar {\bf g}_2 = (1-\epsilon){\bf g}_2$ with some $\epsilon \in [0,1)$, and 
(ii) the solution (${\bf g}_1$, ${\bf F}_1$, $\bar {\bf g}_2$, $\bar {\bf F}_2$) results in an objective value smaller than or equal to that with (${\bf g}_1$, ${\bf F}_1$, ${\bf g}_2$, ${\bf F}_2$).


We will show the first statement (i).
Since (${\bf g}_1$, ${\bf F}_1$, ${\bf g}_2$, ${\bf F}_2$) is a feasible solution, it 
satisfies the constraints for ${\rm SNR}_1$ and ${\rm SNR}_2$ in \eqref{opt:con}.
First, for ${\rm SNR}_2$, using \eqref{eq:SNR_opt} and \eqref{eq:Q2}, we get
\begin{align}
    {\rm SNR}_2 = \eta_2 &= {\bf g}_2^\top ({\bf F}_2 {\bf F}_2^\top \sigma_1^2 + \sigma_2^2{\bf I} )^{-1} {\bf g}_2 
  \nonumber   \\&\le {\bf g}_2^\top (\bar {\bf F}_2 \bar {\bf F}_2^\top \sigma_1^2 + \sigma_2^2{\bf I} )^{-1} {\bf g}_2.
  \label{lemma2:eq:SNR2}
\end{align}
In \eqref{lemma2:eq:SNR2},
we can always choose $\bar {\bf g}_2 = (1-\epsilon){\bf g}_2$ with $\epsilon \in [0,1)$ that satisfies  $\bar {\bf g}_2^\top (\bar {\bf F}_2 \bar {\bf F}_2^\top \sigma_1^2 + \sigma_2^2{\bf I} )^{-1} \bar {\bf g}_2 = \eta_2$. 
This implies that (${\bf g}_1$, ${\bf F}_1$, $\bar {\bf g}_2$, $\bar {\bf F}_2$) satisfies the constraint for ${\rm SNR}_2$.
The constraint for ${\rm SNR}_1$ is also satisfied with (${\bf g}_1$, ${\bf F}_1$, $\bar {\bf g}_2$, $\bar {\bf F}_2$)
since ${\rm SNR}_1$ relies on ${\bf Q}_1 $ in \eqref{eq:Q1} and we have ${\bf F}_1 {\bf F}_2 = {\bf F}_1 \bar {\bf F}_2$.
Therefore, (${\bf g}_1$, ${\bf F}_1$, $\bar {\bf g}_2$, $\bar {\bf F}_2$) is a feasible solution to $\bm{\mathcal{P}}$.

We then will show the second statement (ii).
First, (${\bf g}_1$, ${\bf F}_1$, $\bar {\bf g}_2$, $\bar {\bf F}_2$) yields
a smaller or an equal transmit power of $\mathbb{E}\big[\|{\bf x}_2 \|^2\big]$  
since
\begin{align}
    \mathbb{E}\big[\|{\bf x}_2 \|^2\big] & = \| ({\bf I} + {\bf F}_2 {\bf F}_1 ) {\bf g}_2\|^2 + \|{\bf F}_2{\bf g}_1\|^2
    \nonumber
    \\ 
    & \hspace{3mm} + 
    \| {\bf F}_2( {\bf I} + {\bf F}_1{\bf F}_2 ) \|_F^2 \sigma_1^2
    + \| {\bf F}_2{\bf F}_1 \|_F^2 \sigma_2^2  
      \nonumber  \\
    & 
    \ge \| ({\bf I} + \bar {\bf F}_2 {\bf F}_1 ) \bar {\bf g}_2\|^2 + \|\bar {\bf F}_2{\bf g}_1\|^2
    \nonumber
    \\ 
    & \hspace{3mm} +
    \| \bar {\bf F}_2( {\bf I} + {\bf F}_1 \bar {\bf F}_2 ) \|_F^2 \sigma_1^2
    + \| \bar {\bf F}_2{\bf F}_1 \|_F^2 \sigma_2^2  .
\end{align}
Note that $\mathbb{E}\big[\|{\bf x}_1 \|^2\big]$ in \eqref{eq:power1} are not dependent on $f_{2,N}$ since ${\bf F}_1 {\bf F}_2$ does not include $f_{2,N}$.
Therefore, when $f_{2,N}=0$,
we can always obtain a smaller or an equal objective value of $\bm{\mathcal{P}}$,
while satisfying the constraints in  \eqref{opt:con}.
\end{proof}

We next look into the solution behavior of the message encoding vector for User 2, ${\bf g}_2$. 
To this end, we first formulate the optimization problem $\bm{\mathcal{P}}$ only with respect to ${\bf g}_2$, given by
\begin{align}
    & 
    \underset{{\bf g}_2}
    {\text{min}} & & 
 \alpha \|{\bf F}_1 {\bf g}_2 \|^2 
 + (1-\alpha) \| ({\bf I} + {\bf F}_2 {\bf F}_1 ) {\bf g}_2\|^2
 \nonumber
\\
& \text{s.t.}
& &  
{\bf g}_2^\top {\bf Q}_2^{-1} {\bf g}_2 = \eta_2.
\label{opt:g2}
\end{align}
Defining ${\bf q}_2 = {\bf Q}_2^{-1/2} {\bf g}_2$ where ${\bf Q}_2 = ({\bf Q}_2^{1/2})^2$, we write the equivalent optimization problem~as\footnote{If we conduct the singular value decomposition on ${\bf Q}_2$, we have ${\bf Q}_2 = {\bf U} {\Sigma} {\bf U}^\top $ and obtain
${\bf Q}_2^{1/2} = {\bf U} {\Sigma}^{1/2} {\bf U}^\top$.}
\begin{align}
    & 
    \underset{{\bf q}_2}
    {\text{min}} & & 
    {\bf q}_2^\top {\bf B} {\bf q}_2
 \nonumber
\\
& \text{s.t.}
& &  
\|{\bf q}_2\|^2 = \eta_2,
\label{opt:q2}
\end{align}
where
\begin{align}
    {\bf B} 
    & = \alpha {\bf Q}_2^{1/2}  {\bf F}_1^\top  {\bf F}_1 {\bf Q}_2^{1/2}
    \nonumber
    \\
    & \hspace{.6cm}+ (1-\alpha) {\bf Q}_2^{1/2}  ({\bf I} + {\bf F}_2 {\bf F}_1 )^\top ({\bf I} + {\bf F}_2 {\bf F}_1 )
     {\bf Q}_2^{1/2}.
\end{align}

We then introduce our conjecture on the objective function value of \eqref{opt:q2}, based on which we find the optimal solution for ${\bf g}_2$ in \eqref{opt:g2}.
\begin{conjecture}
    \label{conj1}
    For any ${\bf F}_1$ and ${\bf F}_2$ (in the form of \eqref{eq:F2}),
    \begin{align}
    \min \{\alpha \sigma_1^2, (1-\alpha) \sigma_2^2 \}
    \le \nu_{\min}[{\bf B}]  
    \le (1-\alpha) \sigma_2^2,
    \end{align}
where 
    $\nu_{\min}[{\bf B}] $ denotes the smallest eigenvalue of ${\bf B}$ in \eqref{opt:q2}. 
\end{conjecture}
\begin{proofN3}
We note that $\min \{\alpha \sigma_1^2, (1-\alpha) \sigma_2^2 \} \le (1-\alpha) \sigma_2^2$ for any $\alpha \in (0,1)$.
In the special case with $N=3$,
we will show that $\nu_{\min}[{\bf B}]  = (1-\alpha) \sigma_2^2$ for any ${\bf F}_1$ and ${\bf F}_2$ (in the form of \eqref{eq:F2}).
%
 We first rewrite ${\bf B}  = (1-\alpha) \sigma_2^2 {\bf I} + {\bf C}$ where ${\bf C} = 
 (1-\alpha) \sigma_1^2 {\bf F}_2 {\bf F}_2^\top 
 + {\bf Q}_2^{1/2} \big( 
    \alpha {\bf F}_1^\top {\bf F}_1 
    + (1-\alpha) 
    ( {\bf F}_2 {\bf F}_1 + {\bf F}_1^\top {\bf F}_2^\top + {\bf F}_1^\top {\bf F}_2^\top {\bf F}_2 {\bf F}_1
    )
    \big)
 {\bf Q}_2^{1/2}$. 
 Then,
 showing $\nu_{\min}[{\bf B}] = (1-\alpha) \sigma_2^2$ is equivalent to showing $\nu_{\min}[{\bf C}] = 0$.
 Using
(i) ${\bf Q}_2^{1/2} = \textrm{diag}( [\sigma_2, \sqrt{\sigma_2^2 + f_{2,2} \sigma_1^2}, 0] ) $ from \eqref{eq:SNR2} where $f_{2,3} =0$ from Proposition \ref{pro1}, (ii) $ [{\bf F}_1]_{i,j} = f_{1,i,j}$ for $1 \le j < i \in \{2, 3\}$ while other entries are zeros, and (iii) $ {\bf F}_2 {\bf F}_1 = {\bf 0}$ due to $f_{2,3}=0$, we have
\begin{equation}
{\bf C} = \begin{bmatrix}
c_{11} &  c_{12}  & 0\\
c_{21} & c_{22} & 0 \\
0  & 0 & 0
\end{bmatrix},
\nonumber
\end{equation}
where 
\begin{align}
    c_{11} & = \alpha \sigma_2^2 (f_{1,2,1}^2 + f_{1,3,1}^2),
    \nonumber
    \\
    c_{12} & = \alpha \sigma_2 \sqrt{\sigma_2^2 + f_{2,2}^2 \sigma_1^2}f_{1,3,1} f_{1,3,2},
    \nonumber
    \\
    c_{21} & 
    = \alpha \sigma_2 \sqrt{\sigma_2^2 + f_{2,2}^2 \sigma_1^2}f_{1,3,1} f_{1,3,2},
    \nonumber
    \\
    c_{22} & = (1-\alpha)\sigma_1^2 f_{2,2}^2 + \alpha(\sigma_2^2 + f_{2,2}^2 \sigma_1^2 ) f_{1,3,2}^2.
    \nonumber
\end{align}
We can easily show that, for any ${\bf F}_1$ and ${\bf F}_2$, all the principal minors of ${\bf C}$ (i.e.,  the determinants of the principal matrices) are non-negative, which proves that ${\bf C}$ is positive semi-definite.
We then have $\nu_{\min}[{\bf C}] = 0$, which leads to $\nu_{\min}[{\bf B}] = (1-\alpha)\sigma_2^2$.
\end{proofN3}

We note that, for any $N$, any example that violates
the above conjecture  
has not been observed from the extensive numerical simulations where ${\bf F}_1$ and ${\bf F}_2$ are randomly generated.

\begin{proposition}
\label{pro2}
If Conjecture 1 is true, ${\bf g}_2 = [0, ...,  0, \sqrt{\eta_2}\sigma_2]^\top$ is optimal in  $\bm{\mathcal{P}}$
 when $\alpha \ge  \frac{\sigma_2^2}{\sigma_1^2 + \sigma_2^2}$.
\end{proposition}
\begin{proof}
We have a lower bound of the objective function in \eqref{opt:q2} as
${\bf q}_2^\top {\bf B} {\bf q}_2 
    \ge \nu_{\min}[{\bf B}] \| {\bf q}_2 \|^2.$
From Conjecture~\ref{conj1}, we have $\nu_{\min}[{\bf B}]
= (1-\alpha) \sigma_2^2$
when $\alpha \ge  \frac{\sigma_2^2}{\sigma_1^2 + \sigma_2^2}$.
Then, we have the lower bound as ${\bf q}_2^\top {\bf B} {\bf q}_2 
    \ge (1-\alpha) \sigma_2^2 \| {\bf q}_2 \|^2.$
Here, ${\bf q}_2^\star  = [0, ...,0, \sqrt{\eta_2}]^\top$ satisfies the lower bound 
with $\|{\bf q}^\star_2\|^2 = \eta_2$, 
which can be easily shown by the fact that
all the entries in the last column and row of ${\bf B}$ are  zeros except the last diagonal entry is $(1-\alpha) \sigma_2^2$ due to $f_{2,N}=0$ from Proposition~\ref{pro1}.
In other words, ${\bf q}_2^\star$ is an optimal solution of 
\eqref{opt:q2}.
%
We then have the optimal solution for \eqref{opt:g2} as
${\bf g}_2^\star = {\bf Q}_2^{1/2} {\bf q}_2^\star = [0, ...,0,  \sqrt{\eta_2} \sigma_2]^\top $,
since 
$f_{2,N}=0$ from Proposition 1.
\vspace{-2mm}
\end{proof}
The result of Proposition \ref{pro2}   shows
that it is optimal for User 2 to  transmit the message only over the last channel use when the weight coefficient in \eqref{opt:con} satisfies $\alpha \ge  \frac{\sigma_2^2}{\sigma_1^2 + \sigma_2^2}$.

Using Propositions~\ref{pro1} and ~\ref{pro2}, we next aim to simplify our optimization problem $\bm{\mathcal{P}}$.
In our optimization, we consider the case with $\alpha \ge  \frac{\sigma_2^2}{\sigma_1^2 + \sigma_2^2}$.
From Proposition~\ref{pro2}, we have ${\bf g}_2 = [0, ...,  0, \sqrt{\eta_2}\sigma_2]^\top$ as an optimal solution, which always satisfies ${\rm SNR}_2  = \eta_2$ regardless of other variables.
Thus, we can remove the dependency of the constraint for ${\rm SNR}_2$ in $\bm{\mathcal{P}}$.
Further, to make $\bm{\mathcal{P}}$ more tractable, we define ${\bf q}_1 = {\bf Q}_1^{-1/2} {\bf g}_1$ where ${\bf Q}_1 = ({\bf Q}_1^{1/2})^2$,
which implies 
    ${\rm SNR}_1 = \| {\bf q}_1 \|^2
$ and 
$\| {\bf g}_1 \|^2 = {\bf q}^\top_1 {\bf Q}_1 {\bf q}_1$.
Consequently, we rewrite the transmit powers in \eqref{eq:power1} and \eqref{eq:power2} as
\begin{align}
    \mathbb{E}[\|{\bf x}_1\|^2] 
    &= 
    {\bf q}^\top_1 {\bf Q}_1 {\bf q}_1 +
     \|{\bf F}_1 {\bf F}_2 \|_F^2 \sigma_1^2
    + \| {\bf F}_1 \|_F^2 \sigma_2^2
   \nonumber  \\
    &= 
    \|{\bf q}^\top_1 ({\bf I} + {\bf F}_1{\bf F}_2) \|^2 \sigma_1^2
    + 
    \| {\bf q}^\top_1 {\bf F}_1 \|^2 \sigma_2^2
    \nonumber
    \\
    & \hspace{3mm} + \|{\bf F}_1 {\bf F}_2 \|_F^2 \sigma_1^2
    + \| {\bf F}_1 \|_F^2 \sigma_2^2,
    \label{eq:power1:q1}
    \\
    \mathbb{E}[\|{\bf x}_2\|^2] 
    & =  
    \| {\bf g}_2 \|^2
    + \|{\bf F}_2 {\bf Q}_1^{1/2} {\bf q}_1\|^2
    + 
    \| {\bf F}_2({\bf I} + {\bf F}_1{\bf F}_2 ) \|_F^2 \sigma_1^2
    \nonumber
    \\
    & \hspace{3mm} + \| {\bf F}_2 {\bf F}_1
    \|_F^2 \sigma_2^2.
    \label{eq:power2:q1}
\end{align}
Finally, we simplify our optimization $\bm{\mathcal{P}}$ as
\begin{align}
   (\bm{\mathcal{\widetilde{P}}}): \hspace{0.2cm} & \underset{{\bf q}_1, {\bf F}_1, {\bf F}_2}
    {\text{min}} & & 
 \alpha \mathbb{E}\big[\|{\bf x}_1 \|^2\big] + (1-\alpha) \mathbb{E}\big[\|{\bf x}_2 \|^2\big]
\nonumber
\\
& ~\text{s.t.}
& &  
\| {\bf q}_1\|^2 = \eta_1.
\label{opt:con:3vars}
\end{align}

\section{Iterative Two-Way Optimization for Linear Encoding Schemes in GTWC}

To solve the optimization problem $\bm{\mathcal{\widetilde{P}}}$,
we divide it into two sub-problems, and solve them alternately through a series of iterations.
The first sub-problem is to solve $\bm{\mathcal{\widetilde{P}}}$ for ${\bf q}_1$ and ${\bf F}_1$ given that ${\bf F}_2$ is fixed, and the second sub-problem is to solve  for ${\bf F}_2$ assuming ${\bf q}_1$ and ${\bf F}_1$ are fixed.

\subsection{First sub-problem for obtaining ${\bf q}_1$ and ${\bf F}_1$}

We assume a fixed value for ${\bf F}_2$. 
%
We first show that $\mathbb{E}[\|{\bf x}_2\|^2] $ is upper bounded by sum of the scaled version of $\mathbb{E}[\|{\bf x}_1\|^2] $ and some constant terms as follows:
\begin{align}
    & \mathbb{E}[\|{\bf x}_2\|^2] 
    \overset{(i)}{=}  \| {\bf g}_2 \|^2 + \|{\bf F}_2 {\bf Q}_1^{1/2} {\bf q}_1\|^2
    + 
    \| {\bf F}_2 \|_F^2 \sigma_1^2 
    \nonumber
    \\ 
    & \hspace{1.75cm} + \| {\bf F}_2 {\bf F}_1 {\bf F}_2 \|_F^2 \sigma_1^2
    + \| {\bf F}_2 {\bf F}_1
    \|_F^2 \sigma_2^2
    \nonumber  \\
    & \hspace{.5cm}~~~~~~~  \le \| {\bf g}_2 \|^2 
    + \| {\bf F}_2 \|_F^2 \sigma_1^2 
    \nonumber
    \\
    & \hspace{1cm}
   ~~~~~~ + f_{2,\max}^2 \big( \|{\bf Q}_1^{1/2} {\bf q}_1\|^2 + \| {\bf F}_1 {\bf F}_2\|_F^2 + \| {\bf F}_1 \|_F^2 \big)
   \hspace{-2mm}  \nonumber \\
    & \hspace{.5cm} ~~~~~~~ = \| {\bf g}_2 \|^2 
    + 
    \| {\bf F}_2 \|_F^2 \sigma_1^2 
    + f_{2,\max}^2 \mathbb{E}[\|{\bf x}_1\|^2],
    \label{eq:power2_org}
\end{align}
where $f_{2,{\rm max}}^2 = \underset{i=2,...,N-1}{\max} f_{2,i}^2$. We use the fact that $\textrm{tr}({\bf F}_2 {\bf F}_1 {\bf F}_2 {\bf F}_2^\top) = 0$ to obtain the equality $(i)$ in \eqref{eq:power2_org}.
Accordingly, we upper bound the objective function of $\bm{\mathcal{\widetilde{P}}}$~as
\begin{equation}
    { (1-\alpha) (\| {\bf g}_2 \|^2 
    + \| {\bf F}_2 \|_F^2 \sigma_1^2 )} + 
    (\alpha + f_{2,\max}^2 (1-\alpha)) \mathbb{E}\big[\|{\bf x}_1 \|^2\big],
    \label{eq:Ex2_ineq}
\end{equation}

In the first sub-problem, instead of solving  $\bm{\mathcal{\widetilde{P}}}$ directly, 
we aim to minimize the upper bound of the objective function of $\bm{\mathcal{\widetilde{P}}}$ in \eqref{eq:Ex2_ineq}.
Since the other terms in \eqref{eq:Ex2_ineq} are constants except for $\mathbb{E}\big[\|{\bf x}_1 \|^2\big]$, the first sub-problem is reduced to
\begin{align}
   (\bm{\mathcal{\widetilde{P}}_1}): \hspace{0.2cm} & \underset{{\bf q}_1, {\bf F}_1}
    {\text{min}} & & 
 \mathbb{E}\big[\|{\bf x}_1 \|^2\big]
\nonumber
\\
& \text{s.t.}
& &  
\| {\bf q}_1\|^2 = \eta_1.
\label{opt:con:q1F1}
\end{align}



We will solve $\bm{\mathcal{\widetilde{P}}_1}$
via (i) first obtaining the optimal solution form of ${\bf F}_1$ in terms of ${\bf q}_1$, and then (ii) plugging the optimal solution form of ${\bf F}_1$ in $\mathbb{E}[\|{\bf x}_1\|^2]$ and solving for ${\bf q}_1$.

\textbf{Solving for ${\bf F}_1$.} Note that ${\bf F}_1 \in \mathbb{R}^{N \times N}$ is a strictly lower triangular matrix given by
\begin{equation}
\begingroup 
\setlength\arraycolsep{2.5pt}
{\bf F}_1 = \begin{bmatrix}\setlength\arraycolsep{2pt}
0 &  0  & ... & 0\\
f_{1,2,1} & 0 & ... & 0 \\
\vdots & \ddots & \ddots & \vdots \\
f_{1,N,1}  & ... & f_{1,N,N-1} & 0
\end{bmatrix} = \begin{bmatrix}
0 & 0  & ... & 0\\
{\bf f}_{1,1} & 0 & ... & 0 \\
 & \ddots & \ddots  & \vdots \\
 &  &  {\bf f}_{1,N-1} & 0
\end{bmatrix},
\endgroup
\nonumber
\end{equation}
where ${\bf f}_{1,i} = [f_{1,i+1,i}, f_{1,i+2,i}, ..., f_{1,N,i}]^\top \in \mathbb{R}^{(N-i) \times 1}$, $i\in\{1,...,N-1\}$.
Considering ${\bf q}_{1} = [q_{1,1}, q_{1,2}, ..., q_{1,N}]^\top$, we define the vector that contains a portion of the entries of ${\bf q}_{1}$ as
\begin{equation}\label{eq:h_i}
    {\bf h}_{i} = [q_{1,i+1}, q_{1,i+2}, ..., q_{1,N}]^\top 
    \in \mathbb{R}^{(N-i) \times 1},
\end{equation}
where $i\in \{ 0,...,N-1\}$.
%
With the defined vectors $\{{\bf f}_{1,i}\}$ and $\{{\bf h}_{i}\}$, we can rewrite $\mathbb{E}[\|{\bf x}_1\|^2] $ in \eqref{eq:power1:q1} as
\begin{align}
    \mathbb{E}[\|{\bf x}_1\|^2] 
    =\sum_{i=1}^{N-1} \Phi_i({\bf f}_{1,i}) +  \sigma_1^2 \big( q_{1,N-1}^2 + q_{1,N}^2 \big),
    \label{eq:power1_hf}
\end{align}
where $\Phi_1({\bf f}_{1,1})\triangleq \vert {\bf h}_1^\top {\bf f}_{1,1} \vert^2 \sigma_2^2 
     +   {\bf f}_{1,1}^\top {\bf f}_{1,1} \sigma_2^2$ and 
     $\Phi_i({\bf f}_{1,i})\triangleq \big\vert q_{1,i-1} + f_{2,i} {\bf h}_{i}^\top {\bf f}_{1,i} \big\vert^2 \sigma_1^2 
    + \vert {\bf h}_{i}^\top {\bf f}_{1,i} \vert^2 \sigma_2^2  +{\bf f}_{1,i}^\top {\bf f}_{1,i}  (f_{2,i}^2 \sigma_1^2 + \sigma_2^2)$, 
    $i\in \{2,\cdots,N-1\}$.
    
Using \eqref{eq:power1_hf}, our problem of interest (i.e., $\underset{{\bf F}_1}
    {\text{min}} ~ 
 \mathbb{E}\big[\|{\bf x}_1 \|^2\big])$ can be
decomposed into $N-1$ independent problems each in the form of $\underset{{\bf f}_{1,i}}
    {\text{min}} ~ 
 \Phi_i({\bf f}_{1,i})$,
 $i\in\{1,...,N-1\}$. 
Since each independent problem is convex with respect to ${\bf f}_{1,i}$,
we find ${\bf f}_{1,i}$ optimally
by solving $\frac{\partial \Phi_i({\bf f}_{1,i})}{\partial {\bf f}_{1,i}} = {\bf 0}^\top$.
%
Obviously, we have ${\bf f}_{1,1} = {\bf 0}$. Also, for $i\in\{2,...,N-1\}$, we need to solve
\begin{align}
    & \frac{\partial \Phi_i({\bf f}_{1,i})}{\partial {\bf f}_{1,i}}  = (q_{1,i-1} + f_{2,i}{\bf h}_{i}^\top {\bf f}_{1,i})^\top {\bf h}_{i}^\top f_{2,i}\sigma_1^2
    \nonumber
    \\
    & \hspace{1cm} + ({\bf h}_{i}^\top {\bf f}_{1,i})^\top {\bf h}_{i}^\top \sigma_2^2
    +  (  f_{2,i}^2 \sigma_1^2 + \sigma_2^2) {\bf f}_{1,i}^\top = {\bf 0}^\top.
\end{align}
In order to satisfy the above equality, we need to have
\begin{align}
    (f_{2,i}^2 \sigma_1^2 + \sigma_2^2)({\bf h}_{i} {\bf h}_{i}^\top + {\bf I}){\bf f}_{1,i} = - q_{1,i-1} f_{2,i} \sigma_1^2 {\bf h}_{i}.
\end{align}
Finally, the optimal solution form of ${\bf f}_{1,i}$, $i\in\{2,...,N-1\}$, is given in terms of the entries of ${\bf q}_1$ (encapsulated in ${\bf h}_{i}$ according to \eqref{eq:h_i}) by
\begin{align}
    {\bf f}_{1,i} &= - \frac{q_{1,i-1} f_{2,i} \sigma_1^2}{f_{2,i}^2\sigma_1^2 + \sigma_2^2}
    \big( {\bf h}_{i} {\bf h}_{i}^\top + {\bf I} \big)^{-1} {\bf h}_{i}
   \nonumber  \\
    &\overset{(i)}{=} - \frac{q_{1,i-1} f_{2,i} \sigma_1^2}{f_{2,i}^2 \sigma_1^2 + \sigma_2^2}
    \bigg({\bf I} - \frac{{\bf h}_{i}{\bf h}_{i}^\top}
    {1+ \| {\bf h}_{i}\|^2 } \bigg) {\bf h}_{i}
    \nonumber
\\
    &= - \frac{ f_{2,i}\sigma_1^2}{f_{2,i}^2 \sigma_1^2 + \sigma_2^2}
    \frac{q_{1,i-1}}
    {1+ \| {\bf h}_{i}\|^2 }  {\bf h}_{i},
    \label{eq:f1i_opt}
\end{align}
where the Sherman–Morrison formula is used to obtain equality (i) in \eqref{eq:f1i_opt}.

\textbf{Solving for ${\bf q}_1$.}
Putting the optimal solution of $\{{\bf f}_{1,i}\}_{i=1}^{N-1}$ obtained in \eqref{eq:f1i_opt} back into  \eqref{eq:power1_hf},    
we get
\begin{align}
    & \mathbb{E}[\|{\bf x}_1\|^2] 
    = 
    \sum_{i=2}^{N-1} \bigg[ 
    \bigg( q_{1,i-1} 
     - \frac{ f_{2,i}^2 \sigma_1^2}{f_{2,i}^2\sigma_1^2 + \sigma_2^2}
    \frac{q_{1,i-1} \|{\bf h}_{i}\|^2}
    {1+ \| {\bf h}_{i}\|^2 }  
    \bigg)^2 \sigma_1^2 
    \nonumber
    \\
    & \hspace{.5cm} + \bigg( \frac{ f_{2,i}\sigma_1^2}{f_{2,i}^2\sigma_1^2 + \sigma_2^2}
    \frac{q_{1,i-1} \|{\bf h}_{i}\|^2}
    {1+ \| {\bf h}_{i}\|^2 }  \bigg)^2 \sigma_2^2 
    \nonumber
    \\
    & \hspace{.5cm}
     + \bigg( \frac{ f_{2,i}\sigma_1^2}{f_{2,i}^2 \sigma_1^2 + \sigma_2^2}
    \frac{q_{1,i-1} }
    {1+ \| {\bf h}_{i}\|^2 }  \bigg)^2 \|{\bf h}_{i}\|^2 (f_{2,i}^2\sigma_1^2 + \sigma_2^2) \bigg] \nonumber
      \nonumber  \\
    & \hspace{.5cm}+ \big( q_{1,N-1}^2 + q_{1,N}^2  \big) \sigma_1^2
       \nonumber \\
    & 
    \hspace{1.2cm} = \sum_{i=2}^{N-1}  \frac{\sigma_1^2 q_{1,i-1}^2 \big( f_{2,i}^2\sigma_1^2 + \sigma_2^2(1+ \|{\bf h}_{i}\|^2) \big)}
    {(f_{2,i}^2 \sigma_1^2 + \sigma_2^2)(1+ \|{\bf h}_{i}\|^2)} 
    \nonumber
    \\
    &\hspace{2.2cm} + \big( q_{1,N-1}^2 + q_{1,N}^2  \big) \sigma_1^2.
\end{align}
Then, $\bm{\mathcal{\widetilde{P}}_1}$ is reduced to the following optimization problem:
\begin{align}
    & \underset{ {\bf q}_1 }{\text{minimize}} & & 
    \sum_{i=1}^{N-2}  \frac{\sigma_1^2 q_{1,i}^2 \big(  f_{2,i+1}^2\sigma_1^2 + \sigma_2^2(1+ \|{\bf h}_{i+1}\|^2) \big)}
    {(f_{2,i+1}^2\sigma_1^2 + \sigma_2^2)(1+ \|{\bf h}_{i+1}\|^2)} 
 \nonumber
    \\
    & & & \hspace{1cm} + \big( q_{1,N-1}^2 + q_{1,N}^2  \big) \sigma_1^2
    \nonumber
\\
& \text{s.t.}
& &  
\|{\bf q}_1\|^2  = \eta_1.
\label{eq:con:q1}
\end{align}
%
Defining $x_i = q_{1,i}^2 \ge 0$, we rewrite the objective function in \eqref{eq:con:q1} as
\begin{align}
    & \sum_{i=1}^{N-2}  
    \frac{f_{2,i+1}^2 \sigma_1^4 x_i }{{(f_{2,i+1}^2\sigma_1^2 + \sigma_2^2)(1+ x_{i+2} + ... + x_{N})}} 
    \nonumber
    \\ 
    & \hspace{1cm} + \sum_{i=1}^{N-2} \frac{\sigma_1^2 \sigma_2^2 x_i}{f_{2,i+1}^2\sigma_1^2 + \sigma_2^2}
    + \sigma_1^2 (x_{N-1} + x_{N}),
    \nonumber
\end{align}
and the constraint in \eqref{eq:con:q1} as
    $\sum_{i=1}^{N} x_i  = \eta_1$.

Using the vector form of 
${\bf x} = [x_1, ..., x_N]^\top \in \mathbb{R}^{N \times 1}$, we can formulate the equivalent optimization problem as
\begin{align}
    & \underset{ {\bf x} }{\text{minimize}} & &  \sum_{i=1}^{N-1} \frac{{\bf u}_i^\top {\bf x} } { 1+ {\bf m}_i^\top {\bf x}}
    \nonumber
    \\
& \text{subject to}
& &  
{\bf 1}^\top {\bf x}  = \eta_1, \quad {\bf x} \ge {\bf 0},
\label{opt:con:x}
\end{align}
where ${\bf 1} = [1, ..., 1]^\top \in \mathbb{R}^{N \times 1}$ and ${\bf 0} = [0, ..., 0]^\top \in \mathbb{R}^{N \times 1}$. In~\eqref{opt:con:x}, ${\bf u}_i \in \mathbb{R}^{N \times 1}$ and ${\bf m}_i \in \mathbb{R}^{N \times 1}$, $i \in \{1,..., N-1\}$, are defined as
\begin{align}
    & {\bf u}_i  = \bigg[0, ..., 0, \underbrace{\frac{\vert f_{2,i+1} \vert^2 \sigma_1^4}{\vert f_{2,i+1} \vert^2\sigma_1^2 + \sigma_2^2}}_{ 
    i{\rm-th}}, 0, ..., 0 \bigg]^\top, i \in \{1,..., N-2\},
    \nonumber
    \\
    & {\bf u}_{N-1} = \bigg[ 
    \frac{\sigma_1^2 \sigma_2^2}{\vert f_{2,2} \vert^2\sigma_1^2 + \sigma_2^2}, ..., \frac{\sigma_1^2 \sigma_2^2}{\vert f_{2,N-1} \vert^2\sigma_1^2 + \sigma_2^2}, \sigma_1^2, \sigma_1^2
    \bigg]^\top,
    \nonumber
    \\
    & {\bf m}_i = [0, ..., 0, \underbrace{1}_{i{\rm-th}}, ..., 1]^\top, \quad i \in \{1,..., N-2\},
    \nonumber
    \\
    & {\bf m}_{N-1} = [0, ..., 0]^\top,
    \nonumber
\end{align}
where ${\bf u}_i, {\bf m}_i \ge {\bf 0}$.
The equivalent optimization problem in \eqref{opt:con:x} is a multi-objective linear fractional programming~\cite{freund2001solving}.
We thus can adopt commercial software~\cite{MatlabOTB} to solve this~problem.

\subsection{Second sub-problem for obtaining ${\bf F}_2$}

While fixing ${\bf q}_1$ and ${\bf F}_1$,
we formulate the second sub-problem as
\begin{align}
    & 
    (\bm{\mathcal{\widetilde{P}}_2}): \hspace{0.2cm}  \underset{{\bf F}_2}
    {\text{min}} & & 
 \alpha \mathbb{E}\big[\|{\bf x}_1 \|^2\big] + (1-\alpha) \mathbb{E}\big[\|{\bf x}_2 \|^2\big].
 \label{opt:obj:F2}
\end{align}
We aim to
minimize the objective of $\bm{\mathcal{\widetilde{P}}_2}$ for each $f_{2,i}$, $i\in\{2,...,N-1\}$
by setting the derivative with respect to $f_{2,i}$ equal to zero.
Our methodology would yield a  sub-optimal solution given the non-triviality of the problem $\bm{\mathcal{\widetilde{P}}_2}$.



Considering the expression for $\mathbb{E}[\|{\bf x}_2\|^2]$ in \eqref{eq:power2:q1},
we express each of the terms including ${\bf F}_2$ as a sum of entries of ${\bf F}_2$, i.e., $\{f_{2,i}\}_{i =2 }^{N-1}$. First, revisiting the second term in \eqref{eq:power2:q1}, we obtain
\begin{align}
     \|{\bf F}_2 {\bf Q}_1^{1/2} {\bf q}_1\|^2 & = {\bf q}_1^\top {\bf Q}_1^{1/2} {\bf F}_2^\top {\bf F}_2 {\bf Q}_1^{1/2} {\bf q}_1
     \nonumber
     \\
     & = {\bf p}^\top {\bf F}_2^\top {\bf F}_2 {\bf p}
      = \sum_{i=2}^N p_{i-1}^2 f_{2,i}^2,
     \label{eq:derv_f2_term1}
\end{align}
where we assumed that ${\bf p} \triangleq {\bf Q}_1^{1/2} {\bf q}_1 = [p_1, ..., p_N]^\top$ is fixed for tractability although ${\bf Q}_1^{1/2}$ depends on ${\bf F}_2$.
We then express the third term in \eqref{eq:power2:q1} as
\begin{align}
     & \|{\bf F}_2({\bf I} + {\bf F}_1{\bf F}_2 ) \|_F^2 \sigma_1^2  
     \nonumber
     \\
     & \hspace{1cm} = \sigma_1^2 \sum_{i=2}^N f_{2,i}^2 + \sigma_1^2 \sum_{i=2}^{N-2} \sum_{j=i+1}^{N-1} f_{1,j,i}^2 f_{2,i}^2 f_{2,j+1}^2.
     \label{eq:derv_f2_term2}
\end{align}
Also, the last term in \eqref{eq:power2:q1} can be expressed as
\begin{align}
    \| {\bf F}_2 {\bf F}_1 \|_F^2 \sigma_2^2 
    & = \sigma_2^2 \sum_{i=3}^{N}  f_{2,i}^2 \sum_{j=1}^{i-2} f_{1,i-1,j}^2.
    \label{eq:derv_f2_term3}
\end{align}

Since the derivatives, $\partial \mathbb{E}\big[\|{\bf x}_1 \|^2\big]/\partial f_{2,i} $ and $\partial \mathbb{E}\big[\|{\bf x}_2 \|^2\big]/\partial f_{2,i} $, can be readily derived from
\eqref{eq:power1_hf} and \eqref{eq:power2:q1} using
\eqref{eq:derv_f2_term3},  respectively,
we finally have
\begin{align}
    & \alpha \frac{\partial \mathbb{E}\big[\|{\bf x}_1 \|^2\big]}{\partial f_{2,i}}  + (1-\alpha) \frac{\partial \mathbb{E}\big[\|{\bf x}_2 \|^2\big]}{\partial f_{2,i}} 
    \nonumber
    \\
    & \hspace{3cm} = 
     2 \alpha \sigma_1^2 q_{1,i-1} {\bf h}_i^\top {\bf f}_{1,i}  + c_i f_{2,i},
     \label{eq:2ndsub:deriv}
\end{align}
where
\begin{align}
    & c_i  \triangleq 2 \alpha \sigma_1^2 \big( \vert {\bf h}_i^\top {\bf f}_{1,i} \vert^2    + \| {\bf f}_{1,i} \|^2  \big) +   2 (1-\alpha) p_{i-1}^2 + 2 (1-\alpha) \sigma_1^2 
    \nonumber
    \\
    & \hspace{.5cm} + 2 (1-\alpha) \sigma_1^2  \bigg( \sum_{j=i+1}^{N-1} f_{1,j,i}^2 f_{2,j+1}^2 + \sum_{k=2, \; i \ge 4}^{i-2} f_{1,i-1,k}^2 f_{2,k}^2
    \bigg)
    \nonumber
    \\
    &  \hspace{.5cm}
    + 2 \sigma_2^2 (1-\alpha) \sum_{j=1, \; i \ge 3}^{i-2} f_{1,i-1,j}^2.
    \nonumber
\end{align}
By setting the right-hand equation in \eqref{eq:2ndsub:deriv} to be zero, we obtain the solution for $f_{2,i}$ as
\begin{equation}
    f_{2,i} = -\frac{2 \alpha \sigma_1^2 q_{1,i-1} {\bf h}_i^\top {\bf f}_{1,i}}{c_i}, \quad i\in\{2,...,N-1\}.
    \label{eq:f2i}
\end{equation}

The pseudo-code of our iterative method to solve the overall optimization problem  $\bm{\mathcal{P}}$ is summarized in Algorithm \ref{al:alg}.
We solve the two sub-problems alternatively through a series of \textit{outer} iterations denoted in lines 5-17.
In the \textit{inner} iterations, lines 11-13, we solve the second sub-problem.

 \begin{algorithm}[h]
 \caption{Iterative Two-Way Optimization for Linear Encoding in Gaussian Two-Way Channels}
 \label{al:alg}
 \begin{algorithmic}[1]
 \small
\State \textbf{Input.} 
$N$, $\sigma_1^2$, $\sigma_2^2$, $\eta_1$, $\eta_2$, $\alpha$, $\epsilon$
\State \textbf{Output.} 
${\bf g}_1$, ${\bf F}_1$, ${\bf g}_2$, ${\bf F}_2$
\State Obtain the optimal solution for ${\bf g}_2$ and ${f}_{2,N}$ as
${\bf g}_2 = [0, 0, ..., \sqrt{\eta_2} \sigma_2]^\top$ and ${f}_{2,N} = 0$ from Propositions \ref{pro1} and \ref{pro2}
  \State \multiline{ Randomly generate $\{ f_{2,i} \}_{i=2}^{N-1}$}
  \While {$\vert s_{\rm new} - s_{\rm old} \vert > \epsilon$}
    \State \hspace{-3.8mm}  $\bullet$ \textit{\textbf{Sub-problem 1. Obtain ${\bf g}_1$ and ${\bf F}_1$}} 
    \State \multiline{Solve the problem in \eqref{opt:con:x} for ${\bf x} = [x_1, ..., x_N]^\top$ and obtain $q_{1,i} = \sqrt{x_i}$, $i \in \{1,...,N\}$}
    \State Obtain the columns of ${\bf F}_1$, $\{ {\bf f}_{1,i} \}_{i=1}^{N-1}$, from \eqref{eq:f1i_opt} 
    \State Obtain ${\bf g}_1 = {\bf Q}_1^{1/2} {\bf q}_1$ where ${\bf Q}_1$ is given in \eqref{eq:Q1}
     \State  \hspace{-3.8mm} \textit{$\bullet$ \textbf{Sub-problem 2. Obtain ${\bf F}_2$}} 
    \While {$\vert  \nu_{\rm new} - \nu_{\rm old} \vert > \epsilon$} 
        \State \multiline{ Obtain $f_{2,i}$ sequentially for $i \in \{2,...,N-1\}$ by \eqref{eq:f2i}
        \\
        $\nu_{\rm old} \leftarrow \nu_{\rm new}$
        \\
        Calculate the objective function value $\nu_{\rm new}$ of \eqref{opt:con} with the updated $\{ f_{2,i} \}_{i=2}^{N-1}$}
    \EndWhile
    \State \hspace{-3.8mm} $\bullet$ \textit{\textbf{Update values for stopping criterion}} 
    \State $s_{\rm old} \leftarrow s_{\rm new}$
    \State \multiline{ Calculate the objective function value $s_{\rm new}$ of \eqref{opt:con} with the updated ${\bf g}_1$, ${\bf F}_1$, and ${\bf F}_2$}
     \EndWhile
 \end{algorithmic}
 \end{algorithm}

\section{Numerical Experiments}
\label{sec:sim}

We next present numerical simulations to measure the performance of our proposed two-way optimization scheme. We consider $\sigma_1^2 = 1$, $\sigma_2^2 = 0.5$, and $\eta_1 = \eta_2 = 10$.
For our two-way optimization scheme, we consider 30 different initializations of $\{f_{2,i}\}_{i=2}^{N-1}$ with $f_{2,i} \sim \mathcal{U}(0,1)$, and select the best solution. The threshold for the stopping criterion in Algorithm \ref{al:alg} is $\epsilon = 10^{-3}$.
For performance comparisons, 
we consider two baselines.
The first baseline is the open loop scheme where each user only transmits its own message to the other without employing any feedback scheme. In this case, ${\bf F}_1 = {\bf F}_2 = {\bf 0}$, $\|{\bf g}_1\|^2 =  \eta_1 \sigma_1^2$, and $\|{\bf g}_2\|^2 = \eta_2 \sigma_2^2$.
The second baseline is the one-way optimization method that is especially designed for one-way noisy feedback channels,\footnote{In  two-way channels, two channel uses are needed to receive back the transmit signals at each user, while only a single channel use is needed in one-way channels. Therefore, the feedback scheme for the one-way channels can be applied to the two-way channels by designing the feedback scheme for User 1 over the odd/even-numbered channel uses. For User 2, the message can be transmitted only over the last channel use while the feedback information is conveyed over the even/odd-numbered channel uses without scaling.}
for which we consider the optimization scheme proposed in \cite{agrawal2011iteratively}.

Fig. \ref{fig:sumpower_alpha} depicts the weighted sum of transmit powers of the users under the varying weight  $\alpha$ in  $\bm{\mathcal{P}}$ with $N=7$.
We examine the simulation performances for $\alpha \ge  \frac{\sigma_2^2}{\sigma_1^2 + \sigma_2^2} = 0.33$.
Our proposed two-way optimization enables us to design the encoding schemes of both the users adaptively according to the value of $\alpha$.
Specifically,
for a small $\alpha \leq 0.5$,
the solution inclines toward minimizing the transmit power of User 2 in $\bm{\mathcal{P}}$. In this case, providing the feedback information from User 2 to User 1
may increase the weighted sum of transmit powers severely, which causes User 1 not to use the feedback scheme in the low $\alpha$ regime and thus the performance of our method resembles that of the open loop in Fig. \ref{fig:sumpower_alpha}.
On the other hand, as $\alpha$ increases, the problem  $\bm{\mathcal{P}}$ is more focused on minimizing the transmit power of User 1. In this case,
employing the feedback scheme will be beneficial 
since the feedback scheme allows User 1 to use lower transmit power for satisfying the SNR constraint, while it requires User 2 to use more power for providing the feedback information to User 1.
This causes a significant performance enhancement of our method as compared to baselines in Fig. \ref{fig:sumpower_alpha} upon having higher values of $\alpha$ ($0.5\leq \alpha\leq 1$).

\begin{figure}[t]
    \centering
    \includegraphics[width=\linewidth]{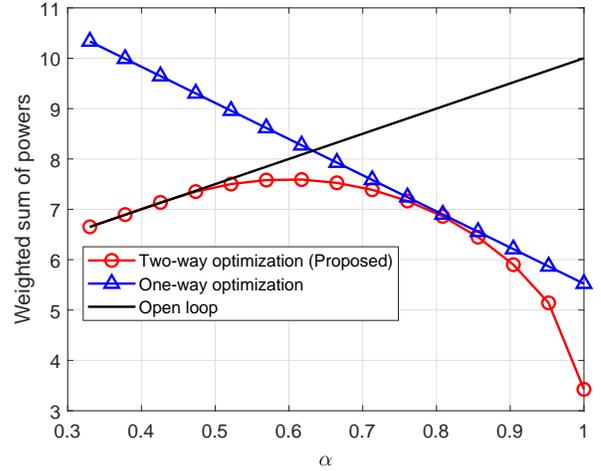}
    \caption{The weighted sum of transmit powers
    along $\alpha$.}
    \label{fig:sumpower_alpha}
\end{figure}


\begin{figure}[t]
    \centering
    \includegraphics[width=\linewidth]{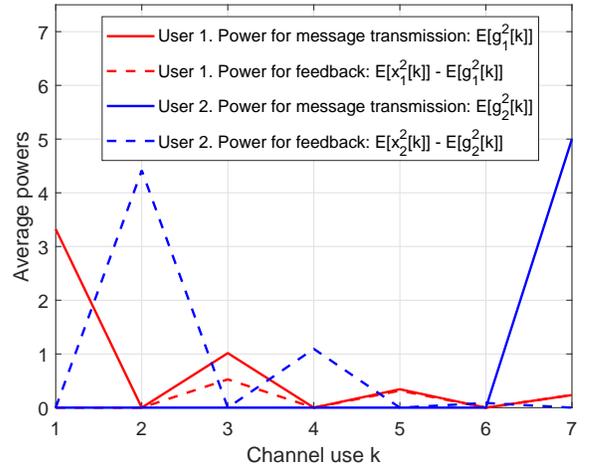}
    \caption{Power profiles for blocklength $N=7$.}
    \label{fig:power_dist}
\end{figure}

Fig. \ref{fig:power_dist} shows the power profiles for the message transmission and the feedback at User 1 and 2 with $N = 7$ and $\alpha=0.8$, which are obtained by our two-way optimization method.
We note that ${\bf g}_i[k]$ is the value of the $k$-th entry of the message encoding vector ${\bf g}_i$, $i \in \{1,2\}$.
From the figure, User 1 only uses the $1,3,5,7$-th channel uses, while User 2 only uses the $2,4,6,7$-th channel uses.
In other words, the channel uses do not overlap between User 1 and 2 except the last channel use.
It is interesting that we do not put any constraints on the separation of the channel usages between the two users when solving the optimization problem in  \eqref{opt:con}. 
However, solving the optimization problem results in the separation of the channel usages.
It can be also seen that
the transmit power of User 1
decreases along the channel uses, which resembles the results of the optimal feedback scheme for one-way noisy feedback channels \cite{agrawal2011iteratively}.
User 2 also exhibits diminishing  power consumption
 along the channel uses.

We note that User 2 conveys its message only over the last channel based on Proposition~\ref{pro2}.
It is worth mentioning that User 2 can split its power for the message transmission over the empty channel uses, i.e., $1,3,5$-th channel uses while maintaining the same objective function value and satisfying the SNR constraints.
This implies that we have multiple optimal solutions for ${\bf g}_2$ given the non-convex structure of the problem $\bm{\mathcal{P}}$. Thus, although in our problem $\bm{\mathcal{P}}$ we are concerned with minimizing the average transmit power over the channel block rather than imposing constraints on the instantaneous transmit powers, 
we may prefer to distribute the powers of ${\bf g}_2$ to mitigate the instantaneous power concentration.




\begin{figure}[t]
    \centering
    \includegraphics[width=\linewidth]{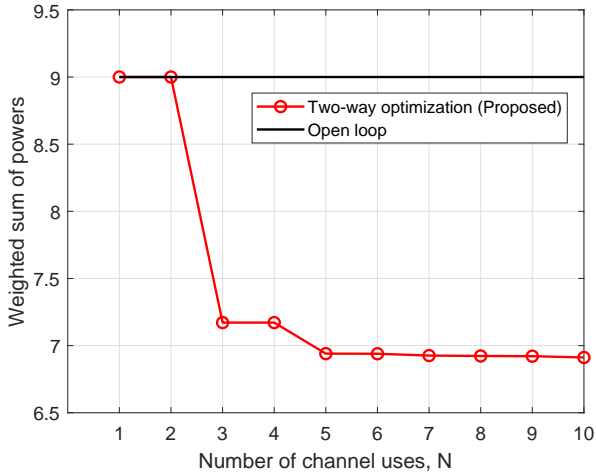}
    \caption{The weighted sum of  transmit powers
    along $N$.}
    \label{fig:sumpower_N}
\end{figure}

Fig. \ref{fig:sumpower_N} depicts the weighted sum of transmit powers under varying number of channel uses $N$ with $\alpha=0.8$. Once $N\geq 3$, User 1 can employ the feedback scheme, which decreases the weight sum of powers significantly.
For $N \ge 5$, the weighted sum of transmit powers of our method is around $23$\% lower than the open loop case.
Also, having larger number of channel uses, i.e., when $N \ge 5$, result in marginal performance gains. 
As a future work, it will be interesting to investigate the performance improvement along $N$ when instantaneous power constraints are imposed so that the users avoid to pour most of their transmit powers to a small portion of the channel uses.


\balance 
\section{Conclusion}


In this work, we presented a system model for linear coding in Gaussian two-way channels 
to bridge the gap between the well-developed model for Gaussian one-way channels proposed by Butman \cite{butman1969general} and the case of Gaussian two-way channels. We then formulated an optimization problem jointly designing the encoding/decoding schemes for the users and investigated its solution behavior. We then proposed an iterative two-way optimization solver to solve our problem.
Through simulations, we showed that our two-way optimization scheme performs better than the non-feedback scheme and the one-way optimization scheme. As a future work, imposing instantaneous transmit power constraints at the users
is worth investigating.

\bibliographystyle{IEEEtran}
\bibliography{ref}

\end{document}